\documentclass[onecolumn,amsthm,9pt]{autart}
\usepackage{amsmath,amssymb,dsfont,graphicx,color,lmodern}
\usepackage{graphicx}               
\usepackage{amsmath,amssymb}
\usepackage[T1]{fontenc}
\usepackage[utf8]{inputenc}

\newtheorem{proposition}[thm]{Proposition}

\newtheorem{remark}{Remark}

\newcommand{\hlf}{{\scriptscriptstyle 1\hspace*{-1pt}/\hspace*{-1pt}2}}

\newcommand{\cS}{{\mathcal{S}}}
\newcommand{\cF}{{\mathcal{F}}}

\newcommand{\interior}[1]{{\kern0pt#1}^{\mathrm{o}}}
\newcommand{\mR}{{\mathbb R}}

\newcommand{\ignore}[1]{}
\newcommand{\ud}{d} 
\newcommand{\W}{{W_2}}

\newcommand{\dbar}{d\hspace*{-0.08em}\bar{}\hspace*{0.1em}}
\newcommand{\lr}[2]{\langle#1,#2\rangle}
\newcounter{rmnum}

\definecolor{gray}{rgb}{0.5,0.5,0.5}

\definecolor{lgrey}{rgb}{0.9,.7,0.7}

\def\spacingset#1{\def\baselinestretch{#1}\small\normalsize}
\spacingset{1}

\begin{document}
	\begin{frontmatter}

\title{Stochastic thermodynamic engines \\under time-varying  temperature profile}
\author[UCI]{Rui Fu}\ead{rfu2@uci.edu},
\author[UCI]{Olga Movilla Miangolarra}\ead{omovilla@uci.edu},
\author[UW]{Amirhossein Taghvaei}\ead{amirtag@uw.edu},
\author[GAT]{Yongxin Chen}\ead{yongchen@gatech.edu},
\author[UCI]{Tryphon T.\ Georgiou}\ead{tryphon@uci.edu}

\address[UCI]{Department of Mechanical and Aerospace Engineering, University of California, Irvine, CA}
\address[UW]{Department of Aeronautics and Astronautics, University of Washington, Seattle, WA 98195}
\address[GAT]{School of Aerospace Engineering, Georgia Institute of Technology, Atlanta, GA}	

\begin{abstract}
In the present paper we study the power output and efficiency of overdamped stochastic thermodynamic engines that are in contact with a heat bath having a temperature that varies periodically with time. This is in contrast to most of the existing literature that considers the Carnot paradigm of alternating contact with heat baths having different fixed temperatures, hot and cold.
Specifically, we consider a periodic and bounded but otherwise arbitrary temperature profile and derive explicit bounds on the power and efficiency achievable by a suitable controlling potential that couples the thermodynamic engine to the external world. Standing assumptions in our analysis are bounds on the norm of the gradient of effective potentials -- in the absence of any such constraint, the physically questionable conclusion of arbitrarily large power can be drawn. 
%
 \end{abstract}


\end{frontmatter}

\section{Introduction}\label{sec:intro}
The classical model of a thermodynamic engine is based on an ensemble
that is brought in contact with two heat baths of different temperatures, hot ($T_h$) and cold ($T_c$), in periodic succession. Based on such a model, Carnot in his foundational treatise \cite{carnot1986reflexions,callen1998thermodynamics} established a limit on the efficiency in transforming heat into mechanical work, that was later on expressed by Thomson (Lord Kelvin) in terms of the ratio of absolute temperatures of the two heat baths, namely $\eta=1-\frac{T_c}{T_h}$, which is known as Carnot efficiency. A salient property in Carnot's analysis was the reversibility of thermodynamic transitions \cite{adkins1983equilibrium}, which necessitated infinitely slow operation, thereby delivering zero power. For more than one hundred years, classical thermodynamics \cite{elwell1972classical,waldram1985theory} did not succeed in addressing questions that pertain to the maximal power that can be delivered during finite-time transitions.

The quest to model far-from-equilbrium thermodynamic transitions and to quantify the maximal power output of thermodynamic engines ultimately led to the development of fluctuation theorems and stochastic thermodynamics \cite{lebon2008understanding,seifert2012stochastic,sekimoto2010stochastic,brockett2017thermodynamics,seifert2008stochastic,schmiedl2007optimal}.
This emerging framework allows studying thermodynamic transitions at the level of a single particle or of small ensembles, and has been applied to the study of biological molecular machines and nano-scale engineering devices. Early work concerning power and efficiency of such microscopic thermodynamic systems focused on Carnot-like heat engines that operate between two thermal heat baths of constant temperature, both in the overdamped \cite{schmiedl2007efficiency} and low friction \cite{dechant2017underdamped} regimes. Recently, this earlier work has been extended to account for continuous and periodic temperature profiles in the linear response regime \cite{bauer2016optimal,brandner2015thermodynamics,Bradner2020geom,fu2020harvesting,frim2022optimal}, the standing assumption being that of small perturbations and linearized dynamics.
Under a  low-friction assumption, explicit bounds on maximal power have also been obtained in the nonlinear regime \cite{miangolarra2021underdamped}.

The aim of this paper is to explore the operation of thermodynamic systems that remain in contact with a single thermal bath of time-varying and periodic temperature profile. 
The work we present below focuses on the fully nonlinear regime, removing the low friction assumption to study overdamped systems. It generalizes our previous work \cite{FU2021109366} in which we studied a Carnot-like heat engine in contact with two heat baths with prespecified hot and cold temperatures. Thus, in the present, we develop formulae for achievable power and efficiency when the thermodynamic engine operates in contact with a heat bath of periodically varying temperature.

The paper is organized as follows. Section  \ref{sec:problem} provides background on stochastic thermodynamics. Section \ref{sec:geometricrepresentationofpower} explains
the conceptualization of a cyclic operation on a thermodynamic manifold, metrized by the Wasserstein metric, and gives geometric expressions for power and dissipation, along with some illustrative examples. 
Section \ref{sec:4} contains the main results on achievable power for  an arbitrary periodic temperature profiles of the thermal bath, under suitable constraints on the Fisher information of the thermodynamic states and the gradient of the controlling potential. 
Section \ref{sec:conclusions} recaps the main ideas and discusses future research.

\section{Background on Stochastic Energetics}
\label{sec:problem}

A basic thermodynamic model that helps quantify the energy exchange between an ensemble of particles in contact with a thermal heat bath and an applied external force is that of overdamped Langevin dynamics
\begin{equation}
	\label{eq:overdamped}
	\gamma \ud  X_t = -\nabla_x U(t,X_t) \ud t + \sqrt{2\gamma k_BT(t)} \ud B_t.
\end{equation}
Here, the stochastic process $X_t\in \mR^n$ represents the position of a single particle with $n$ the dimension of the ambient space,
$\gamma$ the dissipation coefficient, $U(t,x)$ an externally controlled time-varying potential exerting $\nabla_x U(t,X_t)$ force,  $T(t)$ the temperature of a thermal heat bath effected by the stochastic excitation of an $n$-dimensional standard Brownian motion $B_t$, and $k_B$ the Boltzmann constant. Such a model is widely used to model motion of colloidal particles in an ambient heat bath \cite{seifert2012stochastic,peliti2021stochastic}.

The state of the thermodynamic ensemble is identified with  the probability density of the stochastic process $X_t$, denoted by $\rho(t,x)$, and is governed by the Fokker-Planck equation
\begin{equation}	\label{eq: Fokker-Planck}
\frac{\partial \rho}{\partial t} 
= \frac{1}{\gamma} \nabla_x \cdot(\rho (\nabla_x  U(t,x) + k_B T(t) \nabla_x  \log \rho(t,x))).
\end{equation}  
Interestingly, by defining an effective velocity field
\begin{equation}\label{eq:v-def}
v(t,x)= - \frac{1}{\gamma}\left[\nabla_x  U(t,x) + k_B T(t) \nabla_x  \log \rho(t,x)\right],
\end{equation}  
the Fokker-Planck equation takes the form of a  continuity equation $\frac{\partial \rho}{\partial t} +\nabla_x \cdot(\rho v)=0$.

Variations in the potential energy $U(t,X_t)$ of the particles mediate transference of heat and work between the particles and both the heat bath and the externally controlled potential $U(t,\cdot)$. Specifically, the rates of heat and work that are being transferred to an individual particle by the time-varying potential and the heat bath are
\begin{subequations}
\begin{align}\label{eq:heat}
\dbar Q &= \nabla_x U(t,x) \circ \ud X_t,\\
\dbar W &= \frac{\partial U}{\partial t}(t,X_t)\ud t, \label{eq:work}
\end{align}
\end{subequations}
where $\circ$ denotes Stratonovich integration \cite{sekimoto2010stochastic}. The definitions of heat and work ensure validity of the first law of thermodynamics at the level of individual particle, $\ud U(t,X_t)= \dbar W + \dbar Q$. Throughout, $\dbar$ denotes that integration depends on the path (i.e., non-perfect differential, as opposed to $d$). Collectively, at the level of the ensemble, the rates of work and heat infused into the system are
\begin{subequations}
	\begin{align}\label{eq:heat2}
	\dbar \mathcal Q &=  \left[\int_{\mR^n} \langle\nabla_x U(t,x), v(t,X_t) \rangle \rho(t,x) \ud x\right]\ud t,\\\label{eq:work2}
	\dbar \mathcal W &= \left[\int_{\mR^n} \frac{\partial U}{\partial t} (t,x)\rho(t,x) \ud x \right]\ud t,
	\end{align}
\end{subequations}
where $\langle v_1,v_2\rangle=v_1'v_2$ denotes the standard inner product in $\mathbb R^n$,
leading to the first law at the level of the ensemble, $\ud \mathcal E(\rho,U) = 	\dbar \mathcal Q + \dbar \mathcal W $, where
\begin{equation}
\mathcal E(\rho,U) = \int_{\mR^n} U(x)\rho(t, x) \ud x
\end{equation}
is the internal energy of the ensemble system. 

In light of the above, the problem we want to tackle is that of determining the maximal achievable power during a thermodynamic cycle for a temperature profile with period $t_f$. Averaging the work output over a cycle, we obtain that the power is
\begin{equation}
\mathcal P = -\frac{1}{t_f}\int _0^{t_f} \int_{\mR^n} \frac{\partial U}{\partial t}(t,x) \rho(t, x)\ud x \ud t. 
\end{equation}
When the engine operates under periodic driving potential and temperature profile of the same period, the density 
settles to 
a closed periodic orbit,
that is, 
$\rho(0,\cdot) = \rho(t_f,\cdot)$. In general, the potential and temperature profiles can be discontinuous in time, but need to abide by the periodic boundary condition, e.g., $U(0^+,\cdot) = U(t_f^+,\cdot)$ and the same for the temperature. {Throughout we assume that $U(t,x)$ is differentiable in $x$, while both $T(t)$ and  $U(t,x)$ are piecewise differentiable in $t$.}

%

Two additional important thermodynamic quantities that are implicated in the transference of energy along a cycle are the entropy $\mathcal S$ and free energy $\mathcal F$, defined as \cite{owen2012first,parrondo2015thermodynamics}
\begin{subequations}
\begin{align}\label{eq:entropydefination}
\mathcal S (\rho)&=  -k_B\int_{\mR^n} \log(\rho)\,\rho\,\ud x,\\
\label{eq:fenergy}
\mathcal F(\rho,U,T) &= \mathcal E(\rho,U)  - T\mathcal S(\rho).
\end{align}
\end{subequations}
Their relevance will become evident in the sequel.

\section{Cyclic operation: power and dissipation}
\label{sec:geometricrepresentationofpower}

As is evident from the earlier discussion, the state of a thermodynamic ensemble is the probability distribution $\rho(t,\cdot)$ itself, and the cyclic operation under the effect of the heat bath and the driving potential, amounts to traversing a closed orbit
\[
\{\rho(t,\cdot) \mid t\in[0,t_f)\}
\]
over the space of distributions. It is this closed orbit that we seek to optimize, to maximize power output, under the control input that consists of the driving potential 
\[
\{U(t,\cdot) \mid t\in[0,t_f)\}.
\]

A serendipitous connection between thermodynamics and the so-called Wasserstein geometry of probability distributions was discovered by
Aurell etal.\
\cite{aurell2011optimal} (see also \cite{chen2019stochastic}), in that the dissipation along a path of the thermodynamic state can be expressed as the traversed geometric length.
To this end, we briefly describe the basic elements~\cite{villani2003topics}
that are essential for our exposition.

\subsection{Thermodynamic/Wasserstein space $\mathcal P_2(\mR^n)$}
The space of probability distributions on $\mR^n$ with finite second-order moments, denoted by $\mathcal P_2(\mR^n)$, or $\mathcal P_2$ here for short, assumes a Riemannian-like structure. Our interest in this space, as noted, is due to the fact that it serves as state space for our thermodynamic system \eqref{eq: Fokker-Planck} and, in addition, it is normed in a way that the length of trajectories equals the dissipation generated during the corresponding thermodynamic transition of ensembles. Thus,  $\mathcal P_2$ is a natural choice.

The structure of $\mathcal P_2$ is inherited by the so-called Wasserstein metric\footnote{This metric is also known as Monge-Kantorovich, or earth mover's distance. The particular version that we use is based on quadratic transportation cost, noted in the subscript $W_2$.}
\begin{equation}\nonumber
	W_2(\rho_0,\rho_f):=\sqrt{\inf_{\pi \in  \Pi(\rho_0,\rho_f)} \int_{\mR^n\times \mR^n} \|x-y\|^2 \pi(x,y) \ud x \ud y} ,
\end{equation}
for $\rho_0, \rho_f \in \mathcal P_2$. 
The optimization in the above expression is over probability distributions $\pi$ on the product space $\mR^n \times \mR^n$ having $\rho_0$, $\rho_f$ as marginals; $\Pi(\rho_0,\rho_f)$ denotes the set of distributions on the product space with this property, i.e., having the specified marginals.

A tangent displacement $\delta\rho$ about a given density $\rho(\cdot)$ can be identified with a vector field $v(\cdot)$ that effects the infinitesimal perturbation via the continuity equation $\delta\rho=-\left(\nabla_x\cdot(\rho v)\right)\delta t$. It turns out that $v$ can be chosen uniquely by the requirement that it is curl-free, and thus the gradient of suitable (unique modulo a constant) potential, i.e., $v(x)=\nabla_x\phi(x)$ for a scalar function
$\phi$, see \cite[Section 8.1.2, p.\ 246-247]{villani2003topics}.
As a consequence, we can
formally identify  $\delta\rho/\delta t$, $\phi$, $v$, as alternative representations of tangent directions linked bijectively in pairs (modulo a constant in the choice of $\phi$) via the Poisson equation $\delta\rho/\delta t+ \nabla_x\cdot (\rho\nabla_x \phi)=0$ and the curl-free $v=\nabla_x \phi$ requirement.

The tangent space of $\mathcal P_2$ admits the inner product
\[
\langle v_1,v_2\rangle_\rho :=\int_{\mR^n}(\nabla_x\phi_1(x))^\prime \nabla_x\phi_2(x)\rho(x)dx,
\]
that proves to have elegant geometric properties and intrinsic physical significance\footnote{Here, for $i\in\{1,2\}$, $v_i$ and $\phi_i$ correspond via the Poisson equation as explained}.
Firstly, it induces a Riemannian\footnote{As $\mathcal P_2$ can be thought to contain measures, it is often referred to as almost-Riemannian due to the fact that 
vector fields about singular points/measures cannot effect flow in all directions.} structure with norm
\[
\|\frac{\partial \rho}{\partial t}\|^2_\W  :=\int_{\mR^n}\|\nabla_x\phi(x)\|^2\rho(x)dx,
\]
that can be interpreted as (twice the) kinetic energy
of the ensemble (mass$\times$velocity$^2$). Secondly, the minimal length
\[ 
\ell_{\rho_{0:t_f}}:= 
 \int_0^{t_f} \|\frac{\partial \rho}{\partial t}\|_\W \ud t
\] 
of a path traversed by the thermodynamic ensemble between the end-point distributions $\rho_0$ and $\rho_{f}$, is precisely $W_2(\rho_0,\rho_f)$; that is, $W_2(\rho_0,\rho_f)$ is a geodesic distance~\cite[Ch.\ 8]{villani2003topics}. Thirdly, the action integral 
\begin{equation}\label{eq:action}
	\mathcal A_{\rho_{0:t_f}}:= \int_0^{t_f} \|\frac{\partial \rho}{\partial t}\|^2_\W \ud t
\end{equation}
of kinetic energy along any 
thermodynamic transition $\{\rho(t,\cdot); t \in [0,t_f]\}$ turns out to quantify precisely dissipation (entropy production) along the path. Finally, a useful relationship that follows from the Cauchy-Schwartz inequality is that
\begin{equation}
	\mathcal A_{\rho_{0:t_f}} \geq \frac{1}{t_f} \ell_{\rho_{0:t_f}}^2,
\end{equation}
with the bound achieved when the velocity remains constant along the  path ($\|\frac{\partial \rho}{\partial t}\|_\W =\text{constant}$).

\subsection{Energetics over a cycle}
We now consider the overdamped model \eqref{eq:overdamped} and assume that the temperature $T(t)$ is a periodic function of time, independent of the state of the system, with period $t_f$. As before, it is seen to represent the dynamics of any single particle of an ensemble whose distribution $\rho(t,\cdot)$ obeys the Fokker-Planck equation \eqref{eq: Fokker-Planck}. Thence, $\rho(t,\cdot)$ is the state of the ensemble and traverses over a period of duration $t_f$ a closed orbit on the thermodynamic manifold $\mathcal P_2$.

The key in quantifying energy exchange between the system and the environment is the free energy \eqref{eq:fenergy}, which is a function $\mathcal F(\rho(t,\cdot), U(t,\cdot), T(t))$ of the ensemble state,  the potential and the temperature. Since all of the entries are periodic with the same period $t_f$, the change $\Delta \mathcal F$ in the free energy of the system over a cycle is zero, i.e.,
\[
\underbrace{\mathcal F(\rho(t_f,\cdot), U(t_f,\cdot), T(t_f))-\mathcal F(\rho(0,\cdot), U(0,\cdot), T(0))}_{\Delta \mathcal F} = 0,
\]
and hence, we write that $\Delta \mathcal F=
\int_0^{t_f}\dot {\mathcal F}dt=0$.
On the other hand, expanding the rate of change $\dot{\mathcal F}=d\mathcal F/dt$ along the cycle allows separating the contributions of heat and work that come in and out of the ensemble during that time period.
To this, we consider
	\begin{align}
		\frac{d{\mathcal{F}}}{d{t}}=
		&\int\left[
		\frac{\partial{U}}{\partial{t}}\rho+U\frac{\partial{\rho}}{\partial{t}}
		\right]dx+k_BT(t)\int \frac{ \partial{ \rho }}{\partial{t}} \log \rho dx-\dot{T}(t) \mathcal{S}(\rho )\label{eq:first} \\
		=&\int \frac{\partial{U}}{\partial{t}}\rho dx
		+\int \left[\left(U\!+\!k_BT(t)\log \rho \right)\frac{\partial{\rho }}{\partial{t}} \right]dx
		-\dot{T}(t) \mathcal{S}(\rho ),\label{eq:second}
	\end{align}
where for the first equality we used the fact that $\frac{\partial{\rho }}{\partial{t}}$ integrates to zero; spatial integrals from here on are understood as being over $\mR^n$ unless made explicit otherwise. Following \cite{aurell2011optimal} (see also \cite{chen2019stochastic}), the second term in \eqref{eq:second} can be rewritten as
	\begin{align*}
	\hspace*{-1.2cm}&\int \left[\left(U+k_BT(t)\log \rho \right)\frac{\partial{\rho }}{\partial{t}} \right]dx
		=-\int \left(U+k_BT(t)\log  \rho \right)\nabla_x\cdot \left( \rho  v \right)\ud x\\
		=& \int \langle  \nabla_x U + k_B T(t)\nabla_x \log \rho ,   v \rangle   \rho  \;\ud x 
		=-\gamma\int \|v \|^2  \rho  \;\ud x= -\gamma   \|\frac{\partial  \rho }{\partial t}\|_{{\W }}^2,
	\end{align*}
	where the first equality utilizes the Fokker-Planck equation~\eqref{eq: Fokker-Planck}, the second equality follows using integration by parts\footnote{Under the assumption that the controlling potential $U(t,x)$ grows sufficiently fast for large $x$, the state $\rho(t,x)$ vanishes at infinity -- a condition that is needed in integrating by parts.}, the third equality utilizes \eqref{eq:v-def}, and the final equality is a re-write that uses the norm $\|\cdot\|_{{\W }}$ in the tangent space of $\mathcal P_2$.  
	
	From \eqref{eq:second} and \eqref{eq:work2}, the important equation 
	\begin{equation}\label{eq:freeenergycomputation}
	\frac{\ud \cF}{\ud t}\!=\frac{\dbar \mathcal W }{dt}
	-\gamma \|\frac{\partial  \rho }{\partial t}\|_{{\W }}^2-\dot{T}(t) \mathcal{S}( \rho )
	\end{equation}
	follows.
	In this we recognize three contributions. The first term on the right represents work delivered to the system as explained earlier, the second and third terms represent heat exchange with the environment. Of those, the integral along a cycle of the one that involves a quadratic expression of the velocity $\partial \rho/\partial t$, is always negative and vanishes for quasi-static ($t_f\to\infty$) operation. Thus, it represents dissipation, i.e., it represents heat being released to the environment that cannot be recovered in the reverse direction. The last term on the right represents again heat, but this time the flow is reversible with time and the integral over time is independent of the velocity of the ensemble, thus, representing quasi-static heat transference.

We are now in a position to give an expression for the power
\[
\mathcal P:=\frac{1}{t_f}\int_{\rm period}\dbar\mathcal W,
\]
delivered over a cycle.

\begin{proposition}\label{thm1}
	 The power output over a cycle  is
	\begin{align}
\label{eq:power-dissipation-explicit}
		\mathcal P&= -\frac{1}{t_f} \int_0^{t_f}\left[\gamma \|\frac{\partial \rho }{\partial t}\|_{{\W }}^2\right]dt - \frac{1}{t_f} \int_0^{t_f}\left[\dot{T}(t) \mathcal{S}(\rho )\right]\ud t	.		\end{align}														
\end{proposition}
\begin{proof}
In view of \eqref{eq:freeenergycomputation} and the fact that $\Delta\mathcal F=0$ over a cycle, we obtain that work output over a cycle is
	\begin{align*}
	-\int_0^{t_f}\frac{\dbar \mathcal W }{dt}dt &=-\gamma\int_{0}^{t_f}\|\frac{\partial  \rho }{\partial t}\|^2_{{\W }} \ud t -\int_0^{t_f}\dot{T}(t)\mathcal{S}( \rho ) \ud t
	\end{align*}	
	which concludes~\eqref{eq:power-dissipation-explicit}.
	\end{proof}
	
	As noted, the first term is always negative and reprents work that is being dissipated and lost as heat to the environment, this is
	\[
	 \mathcal W_{\text{diss}}:=\gamma\int_{0}^{t_f}\|\frac{\partial  \rho }{\partial t}\|^2_{{\W }} \ud t.
	\]
The average of $-\dot T\mathcal S$ over a period can be both, positive or negative, depending on the control protocol $U$ and relates to the ``useful'' portion of the work that is being extracted (when positive). It is independent of the speed of traversing the cycle, and thereby we refer to it as quasi-static, 
\begin{equation}\label{eq:defWqs}
 \mathcal W_{\text{qs}}:=-\int_0^{t_f}\dot{T}(t)\mathcal{S}( \rho ) \ud t.
\end{equation}
	It can also be expressed in geometric terms as
	\begin{align}
\label{eq:qs}
		\mathcal W_{\rm qs}&=-k_B\int_0^{t_f}T(t)\int  \lr{\nabla_x \log( \rho )}{ v }  \rho \, \ud x \, dt,
		\end{align}										
where we have integrated \eqref{eq:defWqs} by parts and utilized
	\begin{align*}
	\frac{\ud}{\ud t} \cS( \rho )&
=-k_B\int \lr{\nabla_x \log( \rho )}{ v }  \rho  \ud x.
	\end{align*}

To recap, the power delivered over a cycle is
		\begin{equation}
	\mathcal P = \frac{1}{t_f}(\mathcal W_{\rm qs}- \mathcal W_{\text{diss}}).
	\end{equation}
The problem to maximize power by a suitable choice of regulating potential $U(t,\cdot)$, reduces to selecting a closed curve $\{\rho(t,\cdot)\mid t\in[0,t_f]\}$
in the thermodynamic manifold $\mathcal P_2$, such that
\begin{align}\label{eq:max-powerp2}
\mathcal P^\star&:=\max_{\rho(t,\cdot)}~ -\frac{1}{t_f} \int_0^{t_f}\left[\gamma \|\frac{\partial  \rho }{\partial t}\|_{{\W }}^2 +\dot{T}(t) \mathcal{S}( \rho )\right]\ud t,\\\nonumber
&\text{s.t.}\quad \rho(0,\cdot) = \rho(t_f,\cdot).
\end{align}

\begin{remark} \label{eq:physicalinterpretation}
The quasi-static work over a cycle can be written as an area integral as long as the state of the system has a finite dimensional parametrization \cite{miangolarra2022geometry}. On the other hand, the dissipation
 $\mathcal W_{\text{diss}}$
 over a cycle can achieve the lower bound 
  \begin{equation}
  \mathcal W_\text{diss} \geq \frac{\gamma}{t_f} \ell_{\rho_{0:t_f}}^2,
  \end{equation}  
 when the velocity $\|\frac{\partial  \rho }{\partial t }\|_{{\W }}$ is constant  along the curve that the system traverses on $\mathcal P_2$. Thus, assuming constant velocity, the maximal power is
  \begin{equation}
 \mathcal P = \frac{1}{t_f}\left(\mathcal W_\text{qs}- \frac{\gamma}{t_f} \ell_{\rho_{0:t_f}}^2\right).
  \end{equation} 
  In the quasi-static limit, as $t_f \to \infty$, the contribution from  dissipation converges to zero.  $\Box$ \end{remark}
 
  The traditional definition of efficiency, where the work generated is compared to the heat drawn out of the heat bath of highest temperature, does not apply in our present case of a single heat bath of piecewise continuous temperature profile; heat is drawn out of the heat bath at different temperature as this continuous to fluctuates over the period.  The discussion in Remark \ref{eq:physicalinterpretation} motivates the following definition for the efficiency of a thermodynamic cycle in the present context: 
\begin{equation} \label{eq:eff}
\eta:= \frac{\mathcal W_\text{qs} -  \mathcal W_\text{diss} }{\mathcal W_\text{qs} }.
\end{equation}
It is readily seen that the efficiency is always $\eta\leq1$, with equality being achieved in quasi-static limit, when the dissipation is zero.

\subsection{Illustrative examples} 
We next discuss two special cases for which the expression for power can be made explicit. This will not only prove useful later on, but also shed some light into the problem of maximizing power. 

\subsubsection{Carnot-like cycle}
Consider the one-dimensional system \eqref{eq:overdamped} of overdamped particles and assume that the temperature is piecewise constant. That is, we consider Carnot-like operating conditions where the system is brought in contact, alternatingly, with two heat baths  of different temperatures, in which case,
\begin{align}
\label{eq:Carnot}
T(t) = \begin{cases}
T_h,\quad t \in(0,t_\hlf)\\
T_c, \quad t \in (t_\hlf,t_f),
\end{cases}
\end{align}
over a period $t_f$, with $t_\hlf$ to be determined.
The power delivered becomes
\begin{align*} \mathcal P=-\frac{1}{t_f}\left(\gamma \int_{0}^{t_f} \|\frac{\partial  \rho }{\partial t}\|^2_{{\W }} \ud t-\Delta T(\mathcal{S}(\rho_\hlf)-\mathcal{S}(\rho_0)) \right),
\end{align*}
where $\Delta T:=T_h-T_c$, and the thermodynamic state is $\rho_0(\cdot)=\rho(t_0,\cdot)$ and $\rho_\hlf(\cdot)=\rho(t_\hlf,\cdot)$,  at times $t=0$ and $t=t_\hlf$, respectively.  Maximizing power over a choice of control potential $U(t,\cdot)$ and $t_\hlf$, gives  $t_\hlf=t_f/2$ and
\begin{align}
\mathcal P^\star=-\frac{\gamma}{t_\hlf(t_f-t_\hlf )}W_2(\rho_0,\rho_\hlf)^2+\frac{1}{t_f}\Delta T \Delta \cS,\label{eq:wrong?}
\end{align}
where $\Delta \cS=\cS(\rho_\hlf)-\cS(\rho_0)$, see~\cite{FU2021109366} for more details.

 Assuming that the control potential $U(t,x)$ proves sufficient to localize the
thermodynamic state at some point in time, e.g., at $t=0$ where $\rho(0,\cdot)$ can be set to be close to a Dirac delta, then $\cS(\rho_0) \approx -\infty$.
As a consequence, $\Delta \cS \approx \infty$, limitless power can be drawn as $\mathcal P$ in \eqref{eq:wrong?} is not bounded from above. This phenomenon is not particular to the Carnot cycle and can be traced to unreasonable demands on $\nabla U$ to bring the thermodynamic state to a very low entropy condition. Below we highlight that the same is true for a quadratic potential when the thermodynamic state remains Gaussian.

\subsubsection{Gaussian states:}
We consider once again the dynamics in \eqref{eq:overdamped} of overdamped particles subject now to a quadratic controlling potential, and further specialize to one degree of freedom, i.e., $n=1$ and $x\in\mR$. We assume that the thermodynamic state $\rho(t,\cdot)$  of the ensemble is Gaussian with mean zero and variance $\sigma(t)^2$, i.e., $\rho(t,\cdot)=N(0,\sigma(t)^2)$.

If $q(t)$ denotes the ``spring constant'' of the potential, i.e., $U(t,x)=\frac{1}{2}q(t)x^2$, $\sigma(t)$ is governed by the Lyapunov equation 
\begin{align}\label{eq:lyap-gaussian}
\gamma \frac{\ud}{\ud t}(\sigma(t)^2) = -2q(t)\sigma(t)^2 + 2 k_BT(t).
\end{align}
The effective velocity field is
\begin{equation*}
v(t,x) = - \frac{1}{\gamma}(q(t)x  - k_BT(t)\frac{x}{\sigma(t)^2}  ) = \frac{\dot{\sigma}(t)}{\sigma(t)}x,
\end{equation*}
where the last identity follows from Lyapunov equation. Hence, 
\begin{equation*}
\|\frac{\partial \rho}{\partial t}\|_\W^2 =  \frac{1}{\gamma^2}\left(q(t)\sigma(t) - \frac{k_BT(t)}{\sigma(t)}\right)^2 = \gamma  \dot{\sigma}(t)^2.
\end{equation*}
In addition, the entropy of the Gaussian distribution is
\begin{equation}
\cS(\rho) = \frac{k_B}{2} \log(2\pi e \sigma(t)^2).
\end{equation}

Using these identities the expression for the power delivered~\eqref{eq:power-dissipation-explicit} simplifies to
  \begin{align}\label{eq:power-Gaussian}
\mathcal P=-\frac{1}{t_f}
\int_{0}^{t_f}
\left[
\gamma\dot{\sigma}(t)^2+k_B\dot{T}(t)\log(\sigma(t)) \right]\ud t.
\end{align}
It is now intuitively clear that, as long as the temperature does not remain constant (and thus, there is nontrivial temperature gradient), one can apply a similar strategy as the one utilized in the Carnot-like cycle and extract unbounded power from this system. Specifically, with a suitable control protocol $\sigma(t)$ can be made arbitrarily small at the time when the temperature of the heat bath is at its lowest, thereby allowing unbounded power to be extracted as the heat bath moves to higher temperatures.
For completeness, we detail such a strategy in the appendix (see Section \ref{app:1}).

It is apparent that without any restrictions on the allowable control $U(t,x)$, infinite power can be drawn from a temperature varying heat bath. 
The reason is that the quasi-static instantaneous power $k_B\dot T \log(\sigma)$ can increase without bound for a choice of control $U$ that drives the ensemble to a low entropy state, e.g., close to a Dirac. 
While this is highly desirable, it is physically unreasonable. Large gradients of $U(t,x)$ that are needed to localize the thermodynamic state amount to excessive forces being applied to the particles of the ensemble.
Thus, on physical grounds it is reasonable to impose suitable constraints on the gradients of $U$ or $\rho$, along any thermodynamic transition. This is done next.

\section{Maximizing power under constraints}\label{sec:4}

It is insightful to revisit the Fokker Planck equation \eqref{eq: Fokker-Planck}, written as the continuity equation with a vector field $v=\nabla_x \phi$,
\[
\frac{\partial \rho}{\partial t} +\nabla_x \cdot(\rho v)=0
\]
where the ``effective'' potential $\phi$ is composed of two terms $\phi= \frac{1}{\gamma}(R-U)$, with $R$ the probabilistic potential
\[
R(t,x)= - k_BT(t)\log(\rho(t,x)).
\]

The gradient of $U$ represents a physical force that drives the ensemble.
On the other hand, large values for the gradient of $R$ (seen as some sort of  entropic force) also seem physically unrealistic in the context of overdamped (colloidal) particles, where diffusion dominates inertial effects. Below we proceed by postulating and imposing suitable quadratic bounds, and explore the consequences with regard to maximizing power.

Specifically, we postulate that the control mechanism that generates $U(t,x)$, tasked to steer the thermodynamic system along a closed orbit, is restricted in its ability to generate forces. Similarly, that it is restricted in its ability to localize the state to approximate a Dirac. Either of these reasonable constraints can be cast as bounds on the size of the gradients, such as \begin{align}\label{eq:bound1}
&\int_{\mR^n} \|\nabla_x R\|^2\rho dx, \mbox{ or }
\int_{\mR^n} \|\nabla_x U\|^2\rho dx. 
\end{align}
We note that the first of the two square-norms relates directly to the Fisher information of the thermodynamic state
\begin{align}  
I( \rho ):=\int_{\mR^n} \|\nabla_x \log( \rho )\|^2 \rho  \ud x,
\end{align}
since
$
(k_BT(t))^2 I(\rho)=\int 
 \|\nabla_x R\|^2\rho dx
$. Therefore, bounding the $L_2$ norm of $\nabla_x R$ is equivalent to bounding the Fisher information (as long as the temperature is finite and nonzero)
and it is the latter that we will impose in the sequel.
Thus, in the next two subsections, we develop optimal protocols for generating power in thermodynamic systems under suitable bounds.
In passing, we also recall that the quadratic expression
\[
\int_{\mR^n} \|\nabla_x R-\nabla_x U\|^2\rho dx = \gamma   \|\frac{\partial  \rho }{\partial t}\|_{{\W }}^2,
\]
quantifies dissipation, as seen earlier.

\subsection{Maximal power with bounded Fisher information}
\label{sec:constraintonfisher}
We determine an expression for the maximal power that can be extracted under the assumption that
\begin{align} \label{eq:fisherconstraint}
	I( \rho ) \leq I_{\max},
\end{align}
or, equivalently, a corresponding bound on $\int\|\nabla_xR\|^2\rho dx$.
As before, we consider the optimization problem \eqref{eq:max-powerp2} for power generated by the over-damped  model~\eqref{eq:overdamped} over a cycle. In what follows,
\begin{align*}
\bar{T}&:=\frac{1}{t_f}\int_0^{t_f}T(t) \ud t,\\
\text{Var}(T(t))&:=\frac{1}{ t_f} \!\!\int_0^{t_f} \left(T(t)\!-\!\bar{T}\right)^2\ud t,
\end{align*}
denote the mean and variance of the temperature profile.
 
\begin{proposition}
	Under the constraint~\eqref{eq:fisherconstraint} on thermodynamic paths $\rho(t,x)$ over a closed cycle,
	 the maximal power expressed in \eqref{eq:max-powerp2} satisfies	\begin{align}\label{eq:maxboundsgeneral}
		\mathcal{P}^\star\leq \frac{k_B^2I_{\max}}{4\gamma} \text{Var}(T(t)).
	\end{align}
\end{proposition}

\begin{proof}
Over one cycle, the change of entropy is
\begin{align}\label{eq:rateentropy}
\int_0^{t_f} \dot{\mathcal{S}}( \rho )\ud t=\mathcal{S}(\rho_{f})-\mathcal{S}(\rho_{0})
=0
\end{align}
due to the periodic conditions.
Multiplying \eqref{eq:rateentropy} by a constant $C$ and adding to the expression~\eqref{eq:power-dissipation-explicit} for power  yields
\begin{align*}
\mathcal P&= -\frac{1}{t_f} \int_0^{t_f}\int \left[\gamma \|v \|^2 + k_B(T(t)-C) v \cdot \nabla_x \log( \rho ) \right] \rho  \ud x\ud t\\
	& = -\frac{1}{t_f} \int_0^{t_f}\int \left[\sqrt{\gamma} v  + \frac{1}{2\sqrt{\gamma}}k_B(T(t)-C) \nabla_x \log( \rho ) \right]^2 \rho  \ud x\ud t
	+\frac{k_B^2}{4\gamma t_f} \int_0^{t_f} \left(T(t)-C\right)^2I( \rho )\ud t	\\
		&\leq \frac{k_B^2I_{\max}}{4\gamma t_f} \int_0^{t_f} \left(T(t)-C\right)^2\ud t,
\end{align*}
where we use the positivity of the first term and $I( \rho )\leq I_{\text{max}}$.
The best bound over all constants $C$ is obtained by letting $C=\bar{T}$  concluding the result.
\end{proof}
\begin{remark}
	The upper bound of the power extracted from one complete cycle under the Fisher constraint  is proportional to the average fluctuations in the temperature profile \eqref{eq:fisherconstraint}. For the Carnot-like temperature profile \eqref{eq:Carnot}, the maximal power satisfies
		\begin{align*}
			\mathcal{P}^\star \leq \frac{k_B^2I_{\text{max}}}{4\gamma} \frac{(T_h-T_c)^2}{4},
			\end{align*}
			where $T_h:=\max_t\{T(t)\}$ and
$T_c:=\min_t\{T(t)\}$ are
 the maximal and minimal temperatures over the cycle,
		which is consistent with the result in \cite{FU2021109366} that deals with piece-wise constant temperature profile and fast adiabatic transitions (Carnot-like). $\Box$
	\end{remark}
	
In the following, we show that the above bound is tight by providing a protocol that achieves  the  upper bound \eqref{eq:maxboundsgeneral} as $t_f\to 0$, that is, in the limit of fast driving \cite{Blaber2021rapid,schmiedl2007efficiency}.
To this end, we
consider 
\[
\rho(t,x)=N(0, \sigma(t)^2),
\]
once again specializing to $n=1$.
The Fisher information is $I(\rho)=\sigma(t)^{-2}$. We
 select the following variance 
	\begin{align} \label{eq:optimalsigmafisher}
		\sigma(t) = \sigma_{\text{min}}\exp\left(\frac{\kappa}{2\gamma} \left(\int_{t_0}^{t}(T(s)-\bar{T})\ud s\right)\right),
	\end{align}
	where $\kappa:=k_B/\sigma_{\rm min}^2$, $\sigma_{\rm min}:=\min_t \sigma(t)$,
and where $t_0$ is selected so that
\begin{equation}
\label{eq:definationt0}
\int_{t_0}^t \left(T(s)-\bar{T}\right) \ud s \geq 0,
\end{equation}
for all $t$. The profile \eqref{eq:optimalsigmafisher}  can be achieved by the quadratic control protocol
$U(t,x)=\frac{1}{2}q(t)x^2$, where
	\begin{align}\label{eq:q0fishercon}
		q(t)
	&=\frac{\kappa \bar T}{2} +\frac{\kappa T(t)}{2} \left(
	2 e^{-\frac{\kappa}{\gamma}(\int_{t_0}^{t}T(s)-\bar{T}\ud s)}-1
	\right),
	\end{align}
	together with $\sigma(t)^2$, satisfy the Lyapunov equation \eqref{eq:lyap-gaussian}. Under these conditions we have the following:


\begin{proposition} 
	\label{thm:max-power-overdampedfisherin}
Under the control protocol	\eqref{eq:q0fishercon},
	\begin{align*}
		\mathcal P \to \frac{k_B^2I_{\text{max}}}{4\gamma}\text{Var}(T(\cdot)), \mbox{ as }	 t_f\to 0.	\end{align*}
\end{proposition}



\begin{proof}
	We first note that $\sigma(t)$ is periodic and satisfies the constraint $I(\rho)\leq I_{\rm max}$, since $\sigma(t)\geq \sigma_{\rm min}$ following the definition of $t_0$.
The expression for power in
 \eqref{eq:power-Gaussian} gives
	\begin{align}
		\label{eq:powerequivalent}
		\mathcal{P}=-\frac{1}{t_f}\int_{0}^{t_f}\left[\gamma\sigma_{\text{min}}^2\dot{r}(t)^2e^{2r(t)}+k_B\dot{T}(t)r(t)\right]\ud t
		=-\frac{k_B }{t_f}\int_{0}^{t_f}\left[\frac{\gamma \sigma_{{\text{min}}}^2}{k_B}\dot{r}(t)^2e^{2r(t)}-T(t)\dot{r}(t)\right]\ud t,
	\end{align}
	for $r$ representing the logarithmic ratio
	\[
	r(t):=\log(\frac{\sigma(t)}{\sigma_{\text{min}}})=  \frac{\kappa}{2\gamma } \int_{t_0}^t\left(T(s)-\bar{T}\right) ds.
	\]
	The equality in \eqref{eq:powerequivalent} follows using integration by parts and the periodic boundary conditions. 
	Then,
	\begin{align} 		
	\mathcal{P} \nonumber
		&\geq -\frac{k_B }{t_f}\int_{0}^{t_f}\left[\frac{\gamma \sigma_{{\text{min}}}^2}{k_B}\dot{r}(t)^2e^{2r_{\text{max}}}-T(t)\dot{r}(t)\right]\ud t\\
		&=\frac{k_B^2}{4\gamma t_f \sigma_{\text{min}}^2}(2-e^{2r_{\text{max}}})\int_{0}^{t_f}\left(T(t)-\bar{T}\right)^2\ud t=(2-e^{2r_\text{max}})\frac{k_B^2 }{4 \gamma \sigma_{\text{min}}^2}\text{Var}(T(t)),\label{eq:lowerbound}
	\end{align}
for
	\begin{align*}
		0\leq r_\text{max}:=\max_t r(t) \leq \frac{\kappa}{2\gamma } \int_{t_0}^t (T(s)-\bar{T}) \ud s\leq \frac{\kappa t_f}{2\gamma }(T_h-T_c).
	\end{align*}
	Thus, as $t_f \to 0$, $e^{2r_\text{max}} \searrow 1$  and the lower bound in \eqref{eq:lowerbound} approaches the upper bound in 
	\eqref{eq:maxboundsgeneral}, which completes the proof.
\end{proof}

To get insight about the tradeoff between efficiency and power, we evaluate the efficiency of the proposed protocol~\eqref{eq:q0fishercon}. Using the expression~\eqref{eq:optimalsigmafisher} for the variance in the definition for quasi-static work~\eqref{eq:defWqs} yields   
	\begin{align*}
	\mathcal W_\text{qs}&=\int_{0}^{t_f} T(t) \dot {\mathcal{S}}( \rho )\ud t
		=
		k_B\int_{0}^{t_f} T(t)\frac{\dot{\sigma}(t)}{\sigma(t)} \ud t=\frac{k_B^2}{2\gamma \sigma_{\text{min}}^2}\int_{0}^{t_f}\!T(t)\left(T(t)-\bar{T}\right)\! \ud t
		=\frac{k_B^2t_f}{2\gamma  \sigma_{\text{min}}^2}\text{Var}(T(t))
	\end{align*}
while \[\mathcal W_\text{qs}-\mathcal{W}_{\text{diss}} =t_f \mathcal P \to   \frac{k_B^2t_f}{4\gamma  \sigma_{\text{min}}^2}\text{Var}(T(t)) \] as $t_f \to 0$ according to the Proposition~\ref{thm:max-power-overdampedfisherin}.  Therefore, the efficiency 
\begin{align*}
    \eta
		&=\frac{\mathcal W_\text{qs}-\mathcal{W}_{\text{diss}}}{\mathcal W_\text{qs}} \to \frac{1}{2}
\end{align*}
as $t_f \to 0$. This is consistent with the observation made in underdamped limit where it is shown that the efficiency at maximum power is $\frac{1}{2}$.

\subsection{Maximal power with an $L_2$-bound on $\nabla_x U(t,x)$}
\label{sec:constraintonpotential}
We now determine an expression for the maximal power that can be extracted under the assumption that
\begin{equation} \label{eq:potentialcon}
	\frac{1}{\gamma} \int \|\nabla_x U(t,x)\|^2  \rho (x) dx \leq M.
\end{equation}
As before, we consider the optimization problem  \eqref{eq:max-powerp2}
and derive the bound on power below.

\begin{proposition}\label{prop:fisher}
	Under the constraint \eqref{eq:potentialcon}, the maximal power in  \eqref{eq:max-powerp2} is bounded as follows,
	\begin{align}\label{eq:P-upper-bound-U}
		\mathcal{P}^\star\leq  \frac{M}{4}\frac{1}{t_f} \int_0^{t_f}\frac{T_h-T(t)}{T(t)}\ud t.
	\end{align}
\end{proposition}
\begin{proof}
	Over one cycle, the rate of change of the entropy is 
	\begin{align*}
	\dot{S}( \rho )  &= -k_B \int v \cdot \nabla_x \log  \rho   \rho  \ud x= \frac{k_B}{\gamma} \int \left[\nabla_x U \cdot \nabla_x \log ( \rho ) + k_B T(t) \|\nabla_x \log( \rho )\|^2\right] \rho  \ud x\\
		&\geq -\frac{k_B}{\gamma}\left(\int \|\nabla_x U\|^2 \rho  \ud x\right)^{\frac{1}{2}}\left(\int \|\nabla_x \log ( \rho )\|^2 \rho  \ud x\right)^{\frac{1}{2}}
		+ \frac{k_B^2}{\gamma}T(t) \int \|\nabla_x \log ( \rho )\|^2 \rho  \ud x\\
		&\geq \frac{k_B}{\gamma} \left[-\sqrt{\gamma M} \sqrt{I( \rho )}  + k_B T(t) I( \rho )\right]\geq -\frac{M}{4T(t)},
	\end{align*} 
	where the first inequality is by Cauchy-Schwartz, the second inequality is due to the constraint~\eqref{eq:potentialcon}, and the last inequality is obtained by minimizing over $I( \rho )$. 
Thus, 
	\begin{align*}
		\mathcal P= \frac{1}{t_f} \int_0^{t_f}\left[-\gamma \|\frac{\partial  \rho }{\partial t}\|_{W_2}^2 +T(t) \dot{S}( \rho )\right]\ud t\leq \frac{1}{t_f} \int_0^{t_f}T(t) \dot{S}( \rho )\ud t,
	\end{align*}  
where we neglected the dissipation term and used integration by parts and cyclic boundary conditions. In order to apply the bound $\dot{S}( \rho )\geq  -\frac{M}{4T(t)}$, we need to ensure $\dot{S}( \rho )$ is multiplied by a negative factor. This is achieved using the periodic boundary condition and subtracting the zero term $-T_h \int_0^{t_f} \dot{S}( \rho )\ud t=0$. Therefore,
\begin{align} \label{eq: maximalpowerboundspotential}
	\mathcal{P}^\star \leq \frac{1}{t_f} \int_0^{t_f}(T(t)-T_h) \dot{S}( \rho )\ud t\leq \frac{M}{4}\frac{1}{t_f} \int_0^{t_f}\frac{T_h-T(t)}{T(t)}\ud t.
\end{align}  
	\end{proof}
\begin{remark}
For a Carnot-like piece-wise constant temperature profile,
\begin{align*} 
	\mathcal{P}^\star \leq \frac{M}{4}\frac{1}{t_f} \left(\int_0^{t_\hlf}\frac{T_h-T_h}{T_h}\ud t+\int_{t_\hlf}^{t_f}\frac{T_h-T_c}{T_c}\ud t\right)=\frac{M}{8}(\frac{T_h}{T_c}-1),
\end{align*}  
 	which is consistent with the result in \cite[Theorem 2]{FU2021109366} that were derived for this special case. 
	 $\Box$ \end{remark}

\begin{remark}
    Unlike the case with constraint on the Fisher information, it is not clear whether the proposed bound is tight. 
\end{remark}

\section{Conclusions}
\label{sec:conclusions}
The present work quantifies the maximal power and the efficiency at maximal power that can be drawn out from a thermodynamic engine in contact with a single heat bath with arbitrary periodic temperature profile. The analysis is carried out in the framework of stochastic thermodynamics \cite{sekimoto2010stochastic}. We motivate two constraints on the control actuation, one on the Fisher information of the thermodynamic states and the other on a quadratic norm of the gradient of the controlling potential. In each case, we obtain insightful bounds for the maximal power and efficiency at maximal power achievable.
Specifically, for the case where the Fisher information is contrained, we show that the maximal power is nearly fully determined by the variance of the temperature profile.  Of particular interest is that the efficiency  at maximal power approaches $\frac{1}{2}$ when the period tends to zero.   An important direction for future work pertains to connections between this observation and the universal bound on efficiency at maximal power output \cite{van2005thermodynamic} being one half of the Carnot efficiency.

\section{Appendix}
\subsection*{Limitless power under arbitrary protocol.}\label{app:1}

We herein explain that,
as long as the temperature does not remain constant
and there are no restrictions on the control potential in \eqref{eq:overdamped},
the power that can be extracted through the ensemble is arbitrarily large. 

	Since $T(t)$ varies over the cycle, one of the following two cases must hold:
	\begin{enumerate}
		\item[i)] There exists an interval $(a,b) \subset [0,t_f]$ such that $\dot{T}(t)>0$ for $t\in (a,b)$.
		\item[ii)] The derivative $\dot T(t)$ is never positive and there exists $t_\hlf  \in [0,t_f]$ where $T(t)$ is discontinuous, with $T(t_\hlf ^-)<T(t_\hlf ^+)$.
		
	\end{enumerate} 
	In either case, the mechanism for extracting arbitrarily large power is similar. It takes advantage of a localized thermodynamic state (with entropy $\simeq -\infty$) at a point of the cycle when the temperature starts increasing.
	
		Let us first consider case i). Let $ \rho $ be Gaussian $N(0,\sigma(t)^2)$ and
		choose $\delta>0$ such that  $(a+\delta,b-\delta) \subset (a,b)$. Then, let 
	$\sigma(t)$ be  according to
	\begin{align*}
		\sigma(t) = \begin{cases}
	\sigma_{\text{max}},\quad t \in(0,a] \cup   (b,t_f]\\
			k_1t+b_1 ,\quad t \in(a,a+\delta]\\
			\sigma_{\text{min}}, \quad t \in (a+\delta,b-\delta]\\
			k_2 t+b_2, \quad t \in (b-\delta,b]
		\end{cases}
	\end{align*}
	where $k_1<0$, $k_2>0$, $b_1$, and $b_2$ are constants such that $\sigma(t)$ is continuous. 
	From the above, $|k_1|,k_2<\sigma_{\rm max}/\delta$.
	The dissipation term in the expression for power \eqref{eq:power-Gaussian} satisfies 
	\begin{align*}
		-\frac{1}{t_f}\int_{0}^{t_f}\gamma \dot{\sigma}(t)^2 dt 
		&>-2\gamma \frac{\sigma_{\text{max}}^2}{t_f\delta}.
	\end{align*}
	The second term of \eqref{eq:power-Gaussian} decomposes into four parts, following $\sigma(t)$, where
	\begin{eqnarray*}
			\hspace{100pt}\int_{(0,a]\cup (b, t_f]}
			k_B \dot{T}(t)\log(\sigma(t)) dt
			&=& 
			k_B(T(a) - T(b)) \log(\sigma_{\text{max}})\\
		  \int_{(a,a+\delta]}
		  k_B \dot{T}(t)\log(\sigma(t)) dt
		  &\leq&   
		  k_B(T(a +\delta) -T(a)) \log(\sigma_{\text{max}})\\
		\int_{(b-\delta,b]}
		k_B \dot{T}(t)\log(\sigma(t)) \ud t&\leq&
		k_B(T(b) -T(b-\delta)) \log(\sigma_{\text{max}}) \\
		\int_{(a+\delta,b-\delta)}
		k_B \dot{T}(t)\log(\sigma(t)) dt&=& k_B(T(b-\delta)-T(a+\delta)) \log(\sigma_{\text{min}})
	\end{eqnarray*}
	where the fact that $\dot{T}(t)>0$ is used for the inequalities.
	Combining the above, we have
	\begin{align}\label{eq:power-bound}
		\mathcal P & \geq -2\gamma\frac{\sigma^2_{\text{max}}}{t_f\delta}-\frac{k_B}{t_f}\Delta T_1 \log(\frac{\sigma_{\text{min}}}{\sigma_{\text{max}}})
		,
	\end{align}
	where $\Delta T_1:=T(b-\delta) - T(a+\delta) >0$.
	  As $\sigma_{{\text{min}}}\to 0$, the lower bound for the power tends to $\infty$.

	For the second case, select $\sigma(t)$ as before where the interval $(a,b)$ and $\delta$ are chosen such that the point of discontinuity $t_{\hlf} \in (a+\delta,b-\delta) \subset (a,b)$ and $\Delta T_2:=T(b)-T(a)>0$. As a result
		\begin{eqnarray*}
			\hspace{100pt}	\int_{(0,a]\cup (b, t_f]}
			k_B \dot{T}(t)\log(\sigma(t)) dt
			&=& 
			k_B(T(a) - T(b)) \log(\sigma_{\text{max}})\\
		  \int_{(a,a+\delta]}
		  k_B \dot{T}(t)\log(\sigma(t)) dt
		  &\leq&   
		  k_B(T(a \!+\!\delta) -T(a)) \log(\sigma_{\text{min}})\\
		\int_{(b-\delta,b]}
		k_B \dot{T}(t)\log(\sigma(t)) \ud t&\leq&
		k_B(T(b) - T(b-\delta)) \log(\sigma_{\text{min}}) \\
		\int_{(a+\delta,b-\delta)}
		k_B \dot{T}(t)\log(\sigma(t)) dt&=& k_B(T(b-\delta)-T(a+\delta)) \log(\sigma_{\text{min}})
	\end{eqnarray*}
	where the fact that $\dot{T}(t)<0$ is used for the inequalities. Combining the terms concludes a bound similar to~\eqref{eq:power-bound}, with $\Delta T_1$ replaced by $\Delta T_2$, which grows to $\infty$ as  $\sigma_{{\text{min}}}\to 0$.

\section*{Acknowledgments}
The research was supported in part by the NSF under grants 1807664, 1839441, 1901599, 1942523, and the AFOSR under FA9550-17-1-0435.

\end{document}